\documentclass[letterpaper, 10 pt, conference]{ieeeconf}  %

\usepackage[T1]{fontenc}
\usepackage{titlesec}
\usepackage{amsmath, amssymb, graphicx, subfigure, enumerate}

\usepackage{amsthm}
\usepackage{alltt} 
\usepackage{tikz}
\usetikzlibrary{calc,decorations.pathmorphing,shapes.geometric}
\usepackage{svg}
\usepackage{graphicx,booktabs}
\usepackage{mathtools}
\usepackage[boxed, noend]{algorithm2e}
\makeatletter
\let\NAT@parse\undefined
\makeatother
\usepackage[hidelinks]{hyperref}
\usepackage{centernot}
\usepackage{bbm}
\usepackage{stmaryrd} %
\usepackage{yfonts}

\theoremstyle{definition}
\newtheorem{theorem}{Theorem}

\newtheorem{definition}{Definition}
\newtheorem{remark}{Remark}
\newtheorem{lemma}{Lemma}
\newtheorem{corollary}{Corollary}
\newtheorem{example}{Example}

\theoremstyle{definition}

\title{\LARGE \bf A Global Coordinate-Free Approach to Invariant Contraction on Homogeneous Manifolds}

\author{Akash Harapanahalli and Samuel Coogan%
\thanks{Akash Harapanahalli and Samuel Coogan are with the School of Electrical and Computer Engineering, Georgia Institute of Technology, Atlanta, GA, USA, 30318. \{\texttt{aharapan},\texttt{sam.coogan}\}\texttt{@gatech.edu}}%
}

\newcommand{\defi}[1]{\emph{#1}}
\newcommand{\ie}{\emph{i.e.}}

\newcommand{\<}{\langle}
\renewcommand{\>}{\rangle}

\newcommand{\so}{\mathfrak{so}}

\newcommand{\Ad}{\operatorname{Ad}}

\newcommand{\conj}{\operatorname{conj}}
\newcommand{\Hor}{\operatorname{Hor}}

\newcommand{\dangle}[1]{\langle\hspace{-0.2em}\langle #1 \rangle\hspace{-0.2em}\rangle}

\newcommand{\R}{\mathbb{R}}

\newcommand{\calA}{\mathcal{A}}
\newcommand{\calB}{\mathcal{B}}

\newcommand{\calL}{\mathcal{L}}

\newcommand{\bbS}{\mathbb{S}}

\newcommand{\bfu}{\mathbf{u}}

\newcommand{\ol}[1]{\overline{#1}}

\newcommand{\olu}{\ol{u}}

\newcommand{\olV}{\ol{V}}

\newcommand{\olX}{\ol{X}}
\newcommand{\olY}{\ol{Y}}

\newcommand{\frakg}{\mathfrak{g}}
\newcommand{\frakh}{\mathfrak{h}}

\newcommand{\frakm}{\mathfrak{m}}

\definecolor{dblue}{rgb}{.098,.243,.424}
\definecolor{dcompb}{RGB}{157,35,0}  %

\begin{document}

\maketitle
\thispagestyle{empty}
\pagestyle{empty}

\begin{abstract}

In this work, we provide a global condition for contraction with respect to an invariant Riemannian metric on reductive homogeneous spaces. 
Using left-invariant frames, vector fields on the manifold are horizontally lifted to the ambient Lie group, where the Levi-Civita connection is globally characterized as a real matrix multiplication.
By linearizing in these left-invariant frames, we characterize contraction using matrix measures on real square matrices, avoiding the use of local charts.
Applying this global condition, we provide a necessary condition for a prescribed subset of the manifold to possibly admit a contracting system, which accounts for the underlying geometry of the invariant metric.
Applied to the sphere, this condition implies that no great circle can be contained in a contraction region. 
Finally, we apply our results to compute reachable sets for an attitude control problem.

\end{abstract}

\section{INTRODUCTION}

Contraction theory provides a powerful suite of tools for analyzing nonlinear systems by analyzing the distance between pairs of trajectories; see~\cite{AminzareSontag:2014, Bullo:CTDS, DavydovBullo:2024} for recent surveys on the rich history of contraction analysis in dynamical systems.
Contraction is usually studied on vector spaces, from two main perspectives: i) after equipping the state space with a Riemannian structure, the (generalized) Demidovich conditions
verify stability of the variational system~\cite{LohmillerSlotine_ContrNLSys:1998,ManchesterSlotine:2017}; 
ii) equipping the state space with a norm, the matrix measure of the linearization verifies contraction~\cite{Sontag_CSWI:2010,DavydovJafarpourBullo_NonEuclidean:2022}.
Some applications of contraction analysis include: 
analysis and design of networked systems~\cite{DeLellisRusso:2010,RussoSontag:2012};
control design on Lie groups using control contraction metrics~\cite{ManchesterSlotine:2017,WuYiManchester_CCMLieGroup:2024}; 
Lyapunov function design for monotone systems~\cite{Coogan:2019}; 
and robustness analysis of implicit neural networks~\cite{JafarpourBullo:2021}. %

There are many existing approaches formulating contraction in the coordinate-free setting. One approach uses tools from affine differential geometry~\cite{SimpsonBullo_ContrRiemann:2014}, where contraction on Riemannian manifolds is characterized using the Levi-Civita connection. 
Another approach uses a globally defined Finsler-Lyapunov function~\cite{ForniSepulcre_DiffLyap:2013}, where contraction on Finsler manifolds is characterized using Lyapunov conditions on the variational system.
Recently, these results have been extended fully into the coordinate-free setting, where the contraction criterion was characterized using the complete lift of the system and a Finsler-Lyapunov function~\cite{WuDuan_GeoLyapCharISSFinsler:2022}. 
In doing so, a converse result showed that every contracting system admits a Finsler-Lyapunov function.
Further, the complete lift has has been used to connect contraction and local exponential stability on Riemannian manifolds~\cite{WuYiRantzer:2024}.

A manifold equipped with a transitive Lie group action is called a homogeneous space, and provides a canonical equivalence between points.
When the relevant structures (\emph{e.g.}, metrics and connections) remain invariant under this equivalence, problems can often be simplified by lifting to the Lie group, providing a separate suite of algebraic tools for global analysis which can improve theoretical~\cite{Nomizu_InvAffConn:1954,GallierQuantance_DiffGeoLGComp:2020} and computational~\cite{Munthe-Kaas_RKMK:1999} capabilities.
In control theory, homogeneous spaces have been studied in the framework of differential positivity~\cite{MostajeranSepulchre_MonHomSpace:2018}, which is the generalization of monotonicity to manifolds~\cite{ForniSepulcre_DiffPosSys:2015}. 
In these settings, a cone field respecting the group action on the manifold improves tractability in verifying stable attractors for differentially positive systems~\cite{MostajeranSepulcre_InvDiffPos:2018}.
Recently, an almost global tracking controller has been proposed for fully actuated mechanical systems on homogeneous spaces~\cite{WeldeKumar:2024}.

\emph{Contributions} \quad
The main contribution of this work is a simplified global condition for contraction with respect to an invariant metric on reductive homogeneous spaces, which we use to prove a necessary condition for the existence of a contracting system.
In Lemma~\ref{lem:well_def_PX}, identifying the dynamics using a left-invariant frame for a horizontal bundle of a Lie group, we linearize the dynamics into a real $m\times m$ matrix at each point along the actual manifold itself, allowing us to compute the Levi-Civita connection with a standard matrix multiplication.
In Lemma~\ref{lem:propPX}, we show how the abstract contraction condition from~\cite{SimpsonBullo_ContrRiemann:2014} can be transformed into a familiar matrix measure computation.
In Theorem~\ref{thm:inv_contr}, we use these facts to provide a global characterization of a contracting system with respect to an invariant metric.
Next, we pose the following question: under what necessary conditions does there exist a contracting system on a given subset $U$?
Clearly, contraction regions should be contractible in the topological sense, but this does not account for the underlying geometry of the space.
Using our contraction condition, we provide another necessary condition which captures the underlying structure when the metric is invariant under a group action.
Specifically, Theorem~\ref{thm:not_global_contr} shows how invariant contraction regions cannot contain the image of circular one-parameter subgroups acting on the manifold. Corollary~\ref{cor:natred_closed_geo} elaborates that contraction regions cannot contain any closed geodesics if the underlying space is naturally reductive (\emph{e.g.}, symmetric). 
As an example, we show how there are no systems on the sphere (with the standard metric) which contract regions containing any great circle.
Finally, we use invariant contraction theory to simplify reachable set computation for an attitude control system on $SO(3)$.

\section{BACKGROUND AND NOTATION}

Let $M$ denote a ($C^\infty$ smooth) manifold, let $T_pM$ denote the tangent space at $p\in M$, and $TM = \bigsqcup_{p\in M} T_pM$ denote the tangent bundle.
Given a smooth map $f:M\to N$ between manifolds $M$ and $N$, 
let $Tf : TM \to TN$ denote the tangent (differential) map, and $T_pf:T_pM \to T_{f(p)}N$ denote its restriction to $T_pM$.
A vector field $X$ on $M$ (denoted $X\in\Gamma^\infty(TM)$) is a smooth section of the tangent bundle.
Given $X\in\Gamma^\infty(TM)$, let $\Phi_X:J\times M\to M$ be the flow, where $t\mapsto\Phi_X^t(p)$ is the maximal integral curve of $X$ containing $t=0$.
For a smooth map $f:M\to\R$ and $X\in\Gamma^\infty(TM)$, let $\calL_Xf$ denote the Lie derivative.
We use $\dangle{\cdot,\cdot}$ to denote a Riemannian metric, where $\dangle{\cdot,\cdot}_p$ is the inner product in $T_pM$.
We use $\nabla$ to denote the Levi-Civita connection unless otherwise specified.

A Lie group $G$ is a manifold with compatible group structure.
Let $e\in G$ denote the identity.
Let $\ell_g:G\to G$, $h\mapsto \ell_g(h) = gh$ for every $g\in G$ denote the left translation map.
Let $\frakg = T_eG$ denote the Lie algebra, and for $X\in\frakg$, let $X^L\in\Gamma^\infty(TG)$ denote the corresponding left-invariant vector field on $G$, \ie, $g\mapsto X^L(g) = T_e\ell_g(X)$.
The space $\frakg$ is endowed with the Lie bracket $[\cdot,\cdot]_\frakg$ from the usual Lie bracket of vector fields, as $[X,Y]_\frakg^L = [X^L,Y^L]$.
Let $\exp:\frakg\to G$ denote the Lie exponential map, \ie, $\exp(X) = \Phi^1_{X^L}(e)$.
Given $g\in G$, let $\conj_g :G\to G$, $g'\mapsto \conj_g(g') = gg'g^{-1}$ denote conjugation by $g$, and let $\Ad_g:\frakg\to\frakg$, $X\mapsto \Ad_g(X)=T_e\conj_g(X)$ denote the adjoint representation.
We use Einstein's summation convention, where repeated indices in a term are summed over, \emph{e.g.}, $a^ix_i = \sum_{i=1}^n a^ix_i$. Let $\delta_{ij}$ denote the Kronecker delta, where $\delta_{ij} = 0$ for every $i\neq j$ and $\delta_{ij} = 1$ when $i=j$.

\section{REDUCTIVE HOMOGENEOUS SPACES AND INVARIANT RIEMANNIAN METRICS}

In this section, we recall several key facts about reductive homogeneous spaces. 
We refer to the following modern sources: \cite[Sec. 23]{GallierQuantance_DiffGeoLGComp:2020} for an in depth discussion on reductive homogeneous spaces, and \cite{Schlarb_CovariantHomogeneous:2024} regarding invariant metrics and affine connections on homogeneous spaces through the lens of horizontal lifts.
We also refer to the original work~\cite{Nomizu_InvAffConn:1954} characterizing invariant affine connections.

\subsection{Homogeneous spaces}

Let $G$ be a connected Lie group and $M$ be a manifold. A transitive (left) action of $G$ on $M$ is a smooth map $\lambda : G\times M \to M$ satisfying the following axioms:
\begin{enumerate}[i)]
    \item $\lambda (e,p) = p$ for every $p\in M$;
    \item $\lambda (g,\lambda(g',p)) = \lambda(gg',p)$ for every $g,g'\in G$, $p\in M$;
    \item For every $p,q\in M$, there exists $g\in G$ such that $\lambda (g,p) = q$ (transitivity). 
\end{enumerate}
$(M,\lambda)$ is called a homogeneous space. 
Fix a point $o\in M$, and let $G_o = \{g\in G : \lambda(g,o) = o\}$ denote its stabilizer, which forms a closed subgroup of $G$.
One can equivalently construct the homogeneous space $M$ by taking the quotient $G/G_o$, with the following equivalence between (left) cosets and points in $M$, $gG_o \iff \lambda (g,o)$.
In fact, given any closed subgroup $H\subseteq G$, the quotient $G/H$ is a homogeneous space under the natural action $\lambda(g',gH) = (g'g)H$.
Thus, we use the notation $M = G/H$ to denote a homogeneous space with action $\lambda$, and consider points $p\in M$ the same as a point $gH$ given a representative $g$ in the corresponding coset.
For $g\in G$, let $\lambda_g : M\to M$, $\lambda_g(p) = \lambda(g,p)$ denote the associated diffeomorphism. Let $\pi:G\to G/H$, $g\mapsto\pi(g)=gH$ denote the canonical projection.

\subsection{Horizontal lifts}

An important class of homogeneous spaces arises when the Lie algebra $\frakg$ can be decomposed into the direct sum of two subspaces with some additional structure.
\begin{definition} [Reductive homogeneous space {\cite[Def. 23.8]{GallierQuantance_DiffGeoLGComp:2020}}]
    Let $G$ be a Lie group with Lie algebra $\frakg$ and $H\subseteq G$ denote a closed subgroup with Lie subalgebra $\frakh\subseteq\frakg$. The homogeneous space $G/H$ is called \defi{reductive} if there exists a subspace $\frakm\subseteq\frakg$ such that $\frakg = \frakh \oplus \frakm$, and $\Ad_h(\frakm) \subseteq \frakm$ for every $h\in H$ ($\frakm$ is $\Ad_H$-invariant).
\end{definition}

We follow the notation from~\cite{Schlarb_CovariantHomogeneous:2024}. For a reductive homogeneous space $M=G/H$ with reductive decomposition $\frakg=\frakh\oplus\frakm$, define the \defi{horizontal bundle} $\Hor(G) = \bigsqcup_{g\in G} \Hor(G)_g$, where the fiber $\Hor(G)_g = T_e\ell_g(\frakm)$.
Given vector field $X\in\Gamma^\infty(TM)$, define the \defi{horizontal lift} $\olX\in\Gamma^\infty(\Hor(G))$ as the unique vector field where for $g\in G$,
\begin{align*}
    \olX(g) = (T_g\pi\big|_{\operatorname{Hor(G)}_g})^{-1} X (\pi(g)),
\end{align*}
where $T_g\pi\big|_{\operatorname{Hor(G)}_g}: \Hor(G)_g \to T_{gH}M$ is the invertible restriction of $T_g\pi$.
The horizontal lift takes the vector field $X$ on $G/H$ and defines a vector field $\olX$ on $G$, where $\olX(g)$ lives on the subspace $T_e\ell_g(\frakm) = \Hor(G)_g$, and the projection $T_g\pi(\olX(g))$ is equivalent for any representative $g$ in $gH$.

\subsection{Invariant structures on homogeneous spaces}

In this section, we discuss various invariant geometric structures on homogeneous spaces.
Throughout this section, let $M = G/H$ denote a reductive homogeneous space with reductive decomposition $\frakg = \frakh \oplus \frakm$.

\emph{Metrics}\quad
A Riemannian metric $\dangle{\cdot,\cdot}$ on $M$ is invariant if for every $g\in G$, $p\in M$, and $v_p,w_p\in T_pM$,
\begin{align*}
    \dangle{v_p,w_p}_p = \dangle{T_p\lambda_g(v_p), T_p\lambda_g(w_p)}_{\lambda_g(p)}.
\end{align*}
An inner product $\<\cdot,\cdot\>_\frakm:\frakm\times\frakm\to\R$ is $\Ad_H$-invariant if for every $h\in H$ and $X,Y\in\frakm$,
\begin{align*}
    \<\Ad_h(X), \Ad_h(Y)\> = \<X,Y\>,
\end{align*}
It is known~\cite[Prop. 23.22]{GallierQuantance_DiffGeoLGComp:2020} that there is a one-to-one correspondence between $\Ad_H$-invariant inner products on $\frakm$ and invariant Riemannian metrics on $M$ by requiring the map $T_e\pi|_{\frakm} : \frakm \to T_{eH}M$ to be an isometry.

Given an inner product $\<\cdot,\cdot\>_\frakm$ and a basis $\calA = \{A_1,\dots,A_m\}$ for $\frakm$, we define the inner product $\<\cdot,\cdot\>_\calA$ on $\R^m$, such that $\<v,w\>_\calA = \<v^iA_i,w^jA_j\>_\frakm$.

\emph{Connections}\quad
An affine connection $\nabla$ on $M$ is invariant if for every $g\in G$ and $X,Y\in\Gamma^\infty(TM)$,
\begin{align*}
    \nabla_XY = (\lambda_{g^{-1}})_* (\nabla_{(\lambda_g)_*X} (\lambda_g)_*Y),
\end{align*}
where $(\lambda_g)_* = T\lambda_g \circ X \circ \lambda_{g^{-1}}$ denotes the pushforward.
It was originally shown by Nomizu~\cite[Thm. 8.1]{Nomizu_InvAffConn:1954} that there is a one-to-one correspondence between bilinear maps $\alpha : \frakm\times\frakm \to \frakm$ and invariant affine connections $\nabla$ on $M$.

Following the treatment from~\cite[Thm. 4.15, Def. 4.16]{Schlarb_CovariantHomogeneous:2024}, let $\calA = \{A_1,\dots,A_m\}\subseteq\frakm$ be a basis, $X,Y\in\Gamma^\infty(TM)$, and $\olX,\olY\in\Gamma^\infty(\Hor(G))$ be their horizontal lifts. Since $\{A_1^L(g),\dots,A_m^L(g)\}$ span $\Hor(G)_g$ for every $g$, one can express the horizontal lift as the linear combinations $\olX = \olX^i A_i^L$ and $\olY = \olY^i A_i^L$, for curves $\olX^i : G\to\R$ and $\olY^i:G\to\R$ for every $i = 1,\dots,m$.
As shown in~\cite[Thm 4.15]{Schlarb_CovariantHomogeneous:2024}, the covariant derivative $\nabla^\alpha_XY$ between any two vector fields can be expressed through these horizontal lifts as
\begin{align*}
    \ol{\nabla^\alpha_XY} = (\calL_{\olX} \olY^j) A_j^L + \olX^j\olY^k \alpha_{jk}^L,
\end{align*}
where $\alpha_{jk} = \alpha^i_{jk} A_i = \alpha(A_j,A_k)\in\frakm$ for each $j,k\in\{1,\dots,m\}$, with $\alpha^i_{jk}\in\R$.

\emph{Levi-Civita}\quad
It was originally shown in~\cite[Thm. 13.1]{Nomizu_InvAffConn:1954} that the Levi-Civita connection $\nabla$ associated to an invariant metric $\dangle{\cdot,\cdot}$ induced by $\Ad_H$-invariant inner product $\<\cdot,\cdot\>_\frakm$ is an invariant connection, characterized by the bilinear map
\begin{align} \label{eq:LCalpha}
    \alpha (X,Y) = \tfrac12 [X,Y]_\frakm + U(X,Y),
\end{align}
where $U(X,Y)$ is uniquely determined by
\begin{align} \label{eq:LC_U}
    \<U(X,Y),Z\> = \tfrac12 (\<[X,Z]_\frakm,Y\> + \<X,[Y,Z]_\frakm\>).
\end{align}

In a couple of special cases, the expression for $\alpha$ for the invariant Levi-Civita connection is simplified. 
\begin{definition} [Naturally reductive spaces {\cite[Def. 23.9]{GallierQuantance_DiffGeoLGComp:2020}}]
    Let $G/H$ be a homogeneous space with invariant Riemannian metric $\dangle{\cdot,\cdot}$ induced by $\Ad_H$-invariant inner product $\<\cdot,\cdot\>_\frakm$. $G/H$ is called \defi{naturally reductive} if for every $X,Y,Z\in\frakm$,
    \begin{align*}
        \<[X,Z]_\frakm,Y\> + \<X,[Y,Z]_\frakm\> = 0.
    \end{align*}
\end{definition}
In the case of a naturally reductive homogeneous space, $U(X,Y) = 0$, so $\alpha(X,Y) = \frac12[X,Y]_\frakm$ recovers the Levi-Civita connection.
When $G/H$ is naturally reductive, the geodesics coincide with curves of the form $\pi(g\exp(tX))$, $X\in M$~\cite[Prop. 23.28]{GallierQuantance_DiffGeoLGComp:2020}.
A further subclass of naturally reductive homogeneous spaces are symmetric spaces~\cite[Def. 23.14]{GallierQuantance_DiffGeoLGComp:2020}. 
In symmetric spaces, $[\frakm,\frakm]_\frakg \subseteq\frakh$, so $\alpha = 0$ recovers the Levi-Civita connection.

\subsection{Examples of reductive Riemannian homogeneous spaces}

\begin{figure}
    \centering
    \includesvg[width=0.7\columnwidth]{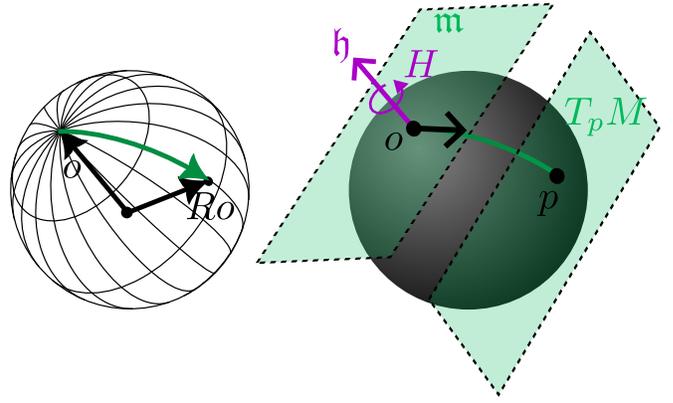}
    \caption{A visualization of the reductive homogeneous space $\bbS^2 = SO(3) / SO(2)$. \textbf{Left:} An element $R\in SO(3)$ acts on $o$, resulting in the vector $p = Ro$. \textbf{Right:} The stabilizer $H= SO(2)$ is the set of rotations $R$ which fix $o$ ($Ro = o$). 
    The element $X\in\frakm$, a $\Ad_H$-invariant subspace, generates the green curve $\pi(\exp(tX))$ and arrives to $p = \pi(\exp(X))$.}
    \label{fig:sphere_homogeneous}
    \vspace{-1em}
\end{figure}

In the sequel, we will consider the following structure.
\begin{definition}
    Let $M = G/H$ be a reductive homogenous space, with reductive decompostion $\frakg = \frakh \oplus \frakm$, $\Ad_H$-invariant inner product $\<\cdot,\cdot\>_\frakm$. 
    We call the tuple $(G,H,\<\cdot,\cdot\>_\frakm)$ a \emph{reductive Riemannian homogeneous space}.
\end{definition}

\begin{example} \label{ex:homogeneous_examples}
We briefly present several examples of reductive Riemannian homogeneous spaces.
\begin{enumerate}
    \item $(V,W^\perp,\<\cdot,\cdot\>|_{W\times W})$---a linear subspace $W\subseteq V$ of inner product space $(V,\<\cdot,\cdot\>)$, as $W = V / W^\perp$, where $W^\perp$ is the orthogonal complement of $W$.
    \item $(G,\{e\},\dangle{\cdot,\cdot})$---a Lie group $G$ with a left-invariant metric. $G$ acts on itself on the left as $\lambda(g_1,g_2) = g_1g_2$. The metric is left-invariant if $\dangle{T_e\ell_g(X),T_e\ell_g(Y)}_{g} = \dangle{X,Y}_e$. Here $H$ is trivial, so $\frakm = \frakg$ and $\dangle{\cdot,\cdot}_e$ is automatically $\Ad_H$-invariant. If the metric is bi-invariant (also invariant under right translation), the space is naturally reductive.
    \item \label{ex:n_sphere} $(SO(n+1),SO(n),\dangle{\cdot,\cdot})$---The $n$-dimensional sphere $\bbS^n$ is diffeomorphic to the quotient $SO(n+1)/SO(n)$. To see this for $\bbS^2$, consider the action of $G=SO(3)$ on vectors of $\R^3$ as rotations (matrix multiplication). Taking $o = [0\ 0\ 1]^T$, the orbit $G\cdot o$ sweeps out the entire unit sphere in $\R^3$, \emph{i.e.}, $\bbS^2$. The stabilizer $G_o$ is the set of rotations about $o$, namely of the form $\begin{bsmallmatrix} \cos(\theta) & -\sin(\theta) & 0 \\ \sin(\theta) & \cos(\theta) & 0 \\ 0 & 0 & 1 \end{bsmallmatrix}$, which is diffeomorphic to $SO(2)$. Taking the metric $\<X,Y\>_\frakg = \tfrac12\operatorname{trace}(X^TY)$, it is easy to show that the restriction $\<\cdot,\cdot\>_\frakm = \<\cdot,\cdot\>_\frakg |_{\frakm \times \frakm}$ is $\Ad_{G_o}$-invariant, giving rise to an invariant metric $\dangle{\cdot,\cdot}$ on $\bbS^2 = SO(3) / SO(2)$ (matching the standard round metric).
    This is visualized in Figure~\ref{fig:sphere_homogeneous}.
    Finally, $\bbS^n$ with this metric is a symmetric space, so $\alpha=0$.
\end{enumerate}
\end{example}

\section{INVARIANT CONTRACTION THEORY}

In~\cite{SimpsonBullo_ContrRiemann:2014}, contraction with respect to a distance defined by a Riemannian metric was formulated in a coordinate-free manner using the Levi-Civita connection.

\begin{definition}[Contracting system {\relax\cite[Def. 2.1]{SimpsonBullo_ContrRiemann:2014}}]
    \label{def:contr_sys}
    Let $M$ be a manifold, $X\in\Gamma^\infty(TM)$, and $U\subseteq M$ be a connected set. Let $\dangle{\cdot,\cdot}$ be a Riemannian metric on $M$.
    We say that $(U,X,\dangle{\cdot,\cdot},c)$ is a \emph{contracting system} on $M$ at rate $c < 0$ if for every $x\in U$ and $v\in T_xM$,
    \begin{align*}
        \dangle{\nabla_{v}X, v}_x \leq c\dangle{v,v}_x,
    \end{align*}
    where $\nabla$ denotes the Levi-Civita connection of $\dangle{\cdot,\cdot}$. 
\end{definition}

Checking this condition in practice, however, requires the use of local charts to verify the generalized Demidovich conditions~\cite[Prop. 2.4]{SimpsonBullo_ContrRiemann:2014}.
In this section, for reductive homogeneous spaces equipped with an invariant metric, we develop a global characterization of a contracting system, using matrix measures on real $m\times m$ dimensional matrices.

\subsection{Linearization in left-invariant frames}

Given a vector field $X\in\Gamma^\infty(TM)$, we first consider the linearization of the horizontal lift $\olX\in\operatorname{Hor}(G)$, with respect to a left-invariant frame from a basis of $\frakm$.

\begin{definition}[Horizontal linearization in left-invariant frame] \label{def:horleftinvlin}
    Let $(G,H,\<\cdot,\cdot\>_\frakm)$ be a reductive Riemannian homogeneous space, and let $\calA=\{A_1,\dots,A_m\}$ be a orthonormal basis for $\frakm$.
    For a vector field $\olX\in\Gamma^\infty(\Hor(G))$, characterized by $\olX = \olX^i A^L_i$ for curves $\olX^i:G\to\R$, define the map $\partial_\calA \olX : G\to\R^{m\times m}$,
    \begin{align*}
        (\partial_\calA \olX)^i_j = \calL_{A_j^L} \olX^i + \olX^k \alpha^i_{jk},
    \end{align*}
    with $\alpha^i_{jk}$ defined as $\alpha(A_j,A_k) = \alpha^i_{jk}A_i$ for the bilinear map $\alpha:\frakm\times\frakm \to \frakm$ from~\eqref{eq:LCalpha}.
\end{definition}

\begin{lemma}[Properties of $\partial_\calA \olX$] \label{lem:well_def_PX}
    Let $(G,H,\<\cdot,\cdot\>_\frakm)$ be a reductive Riemannian homogeneous space and $\calA=\{A_1,\dots,A_m\}$ be an orthonormal basis for $\frakm$.
    Let $X\in\Gamma^\infty(TM)$ and $\olX\in\Gamma^\infty(\Hor(G))$ be its horizontal lift. 
    The following statements hold:
    \begin{enumerate}[i)]
        \item \label{lem:well_def_PX:p1} For any $v\in\R^m$, and $V\in\Gamma^\infty(TM)$ defined by $V(gH) = T_g\pi(v^jA_j^L(g))$,
        \begin{align*}
            \overline{\nabla_{V}X} = (\partial_\calA \olX)^i_j v^j A_i^L = ((\partial_\calA \olX)v)^{\wedge L},
        \end{align*}
        where $\wedge : \R^n \to \frakg$ is the basis expansion.
        \item \label{lem:well_def_PX:p2} For any $g_1,g_2$ such that $g_1H = g_2H$, 
        \begin{align*}
            \partial_\calA\olX(g_1) = \partial_\calA\olX(g_2).
        \end{align*}
    \end{enumerate}
\end{lemma}
\begin{proof}
Let $v\in\R^m$, $g\in G$. Let $\olX^jA_j^L = \olX$ for $\olX^j:G\to\R$, $j=1,\dots,m$. 
Regarding~\eqref{lem:well_def_PX:p1}, since $\olV = v^jA_j^L$,
\begin{align*}
    \overline{\nabla_{V} X} &= (\calL_{v^jA_j^L} X^i) A_i^L + v^jX^k \alpha_{jk}^L\\
    &= v^j(\calL_{A_j^L} \olX^i) A_i^L + v^jX^k \alpha_{jk}^i A_i^L \\
    &= (\calL_{A_j^L} \olX^i + \olX^k \alpha_{jk}^i) v^j A_i^L = (\partial_\calA \olX)_j^i v^j A_i^L.
\end{align*}

Regarding~\eqref{lem:well_def_PX:p2}, let $g_1H = g_2H$. Since $T_{g_1}\pi(\ol{\nabla_{v^jA_j^L} X}(g_1)) = T_{g_2}\pi(\ol{\nabla_{v^jA_j^L} X}(g_2))$,
\begin{align*}
    T_{g_1}\pi ((\partial_\calA \olX(g_1))_j^i v^j A_i^L(g_1)) = T_{g_2}\pi((\partial_\calA \olX(g_2))_j^i v^j A_i^L(g_2)),
\end{align*}
which implies by linearity of $T_g\pi$ that
\begin{align*}
    (\partial_\calA \olX(g_1))_j^i v^j T_{g_1}\pi (A_i^L(g_1)) = (\partial_\calA \olX(g_2))_j^i v^j T_{g_2}\pi(A_i^L(g_2)).
\end{align*}
Finally, since $T_{g_1}\pi (A_i^L(g_1)) = T_{g_2}\pi(A_i^L(g_2))$ by definition of the horizontal lift, $(\partial_\calA \olX(g_1))_j^i v^j = (\partial_\calA \olX(g_2))_j^i v^j$. Since $v$ was arbitrary, $\partial_\calA \olX(g_1) = \partial_\calA \olX(g_2)$.
\end{proof}

Lemma~\ref{lem:well_def_PX} highlights two key important facts about Definition~\ref{def:horleftinvlin}: \eqref{lem:well_def_PX:p1} shows how the invariant affine connection can be computed using a real matrix multiplication; \eqref{lem:well_def_PX:p2} shows that the matrix coincides on every coset in the space $G/H$.

\begin{definition}[Linearization on homogeneous space] \label{def:linhomspace}
    Let $(G,H,\<\cdot,\cdot\>_\frakm)$ be a reductive Riemannian homogeneous space, and let $\calA$ be an orthonormal basis for $\frakm$.
    For a vector field $X\in \Gamma^\infty(T(G/H))$ with horizontal lift $\olX\in\Gamma^\infty(\Hor(G))$, define the map $\partial_\calA X : G/H \to \R^{m\times m}$, 
    \begin{align*}
        \partial_\calA X (gH) = \partial_\calA \olX(g).
    \end{align*}
    By Part~\eqref{lem:well_def_PX:p2} of Lemma~\ref{lem:well_def_PX}, $\partial_\calA X$ is well defined since $\partial_\calA \olX$ coincides on every coset.
\end{definition}

\begin{lemma} \label{lem:propPX}
    Let $(G,H,\<\cdot,\cdot\>_\frakm)$ be a reductive Riemannian homogeneous space, and let $\calA=\{A_1,\dots,A_m\}$ be an orthonormal basis for $\frakm$.
    Let $X\in\Gamma^\infty(TM)$ and $\olX \in\Gamma^\infty(\Hor(G))$ be its horizontal lift.
    Let $\dangle{\cdot,\cdot}$ be the invariant Riemannian metric on $M = G/H$ induced by $\<\cdot,\cdot\>_\frakm$, and let $\nabla$ denote the corresponding Levi-Civita connection.
    For any $g\in G$, $v\in \R^m$, and $V\in\Gamma^\infty(TM)$ such that $V(gH) = T_g\pi(v^jA_j^L(g))$, %
    \begin{align*}
        \dangle{\nabla_{V}X, V}_{gH} &= \<\partial_\calA X(g)^i_j v^j A_i, v^iA_i\>_\frakm \\
        &= \<\partial_\calA X(g) v, v\>_\calA.
    \end{align*}
\end{lemma}
\begin{proof}
    Let $g\in G$ and $v\in\R^m$.
    For notational brevity, let $\partial\olX := \partial_\calA \olX(g)$.
    Applying Part \eqref{lem:well_def_PX:p1} of Lemma~\ref{lem:well_def_PX},
    \begin{align*}
        \dangle{\nabla_VX,V}_{gH} &= \dangle{T_g\pi((\partial \olX)^i_j v^j A_i^L),T_g\pi(v^i A_i^L)}_{gH}.
    \end{align*}
    By invariance of the metric $\dangle{\cdot,\cdot}$,
    \begin{align*}
        \dangle{\nabla_VX,V}_{gH} &= \dangle{T_{gH}\lambda_{g^{-1}} (T_g\pi((\partial \olX)^i_j v^j A_i^L)), \\
        & \quad\quad T_{gH}\lambda_{g^{-1}}(T_g\pi(v^i A_i^L))}_{eH}.
    \end{align*}
    Using the chain rule, $T_{gH}\lambda_{g^{-1}} \circ T_g\pi = T_g (\lambda_{g^{-1}}\circ \pi)$, and $\lambda(g^{-1}, \pi(\cdot)) = \lambda(e,\pi(\ell_{g^{-1}}(\cdot))) = (\pi\circ\ell_{g^{-1}})(\cdot)$. Thus,
    \begin{align*}
        T_{gH}&\lambda_{g^{-1}} (T_g\pi((\partial \olX)^i_j v^j A_i^L)) = T_e\pi(T_{g}\ell_{g^{-1}} ((\partial \olX)^i_j v^j A_i^L)) \\
        &= (\partial \olX)^i_j v^j T_e\pi(T_{g}\ell_{g^{-1}} (A_i^L)) = (\partial \olX)^i_j v^j T_e\pi(A_i^L) \\
        &= (\partial \olX)^i_j v^j A_i,
    \end{align*}
    and similarly, $T_{gH}\lambda_{g^{-1}}(T_g\pi(v^i A_i^L)) = v^i A_i$.
\end{proof}

\begin{remark} \label{rem:mtx_measure_independent}
    Since $\dangle{\nabla_VX,V}_{gH}$ does not depend on the choice of orthonormal basis $\calA$, an immediate corollary of Lemma~\ref{lem:propPX} is that the quantity $\<\partial_\calA X(g) v, v\>_\calA$ is independent of the choice of $\calA$.
\end{remark}

\subsection{Invariant contracting systems}

Using Lemma~\ref{lem:propPX}, we define the matrix measure of the linearization as the maximum over vectors of unit length in any orthonormal basis, similar to the quadratic equality for norms induced by inner products~\cite[Sec. 2.6]{Bullo:CTDS}.

\begin{definition}[Matrix measure]
    Let $(G,H,\<\cdot,\cdot\>_\frakm)$ be a reductive Riemannian homogeneous space.
    For a vector field $X\in\Gamma^\infty(T(G/H))$, define the map $\mu(\partial X(\cdot)):G/H \to \R$,
    \begin{align*}
        \mu(\partial X(p)) = \max_{v\in\R^m, \<v,v\>_\calA = 1} \<\partial_\calA X(p)v,v\>_\calA,
    \end{align*}
    where $\calA$ is any orthonormal basis for $\frakm$. This is well defined by Remark~\ref{rem:mtx_measure_independent}.
\end{definition}

Strictly speaking, $\mu(\partial X(\cdot))$ is notation for a single map, not the composition of $\mu$ and $\partial X$, since $\partial X$ is only defined with respect to a basis. We adopt this notation to connect to the traditional concept of a matrix measure. Once a basis $\calA$ is chosen for $\frakm$, this can be regarded as the composition, since $\partial_\calA X (gH)$ is a real $m\times m$ matrix and $\mu$ is the usual matrix measure.

\begin{definition}[Invariant contracting system] \label{def:inv_contr_sys}
    Let $M = G/H$ be a reductive homogeneous space with reductive decomposition $\frakg = \frakh \oplus \frakm$. 
    Let $X\in\Gamma^\infty(TM)$, $U\subseteq M$ be connected, and $\<\cdot,\cdot\>_\frakm$ be an $\Ad_H$-invariant inner product on $\frakm$. We say that $(U,X,\<\cdot,\cdot\>_\frakm,c)$ is an invariant contracting system on $M$ at rate $c < 0$ if for every $p\in U$, 
    \begin{align*}
        \mu(\partial X(p)) \leq c.
    \end{align*}
\end{definition}

We now present the first Theorem of this work, which connects this condition to the original notion of a contracting system from Definition~\ref{def:contr_sys}.

\begin{theorem} \label{thm:inv_contr}
    Let $M = G/H$ be a reductive homogeneous space with reductive decomposition $\frakg = \frakh \oplus \frakm$. Let $X\in\Gamma^\infty(TM)$, $U\subseteq M$ be connected, $c<0$, $\<\cdot,\cdot\>_\frakm$ be an $\Ad_H$-invariant inner product on $\frakm$, and $\dangle{\cdot,\cdot}$ be the associated invariant metric. The following are equivalent:
    \begin{enumerate}[i)]
        \item $(U,X,\dangle{\cdot,\cdot},c)$ is a contracting system on $M$;
        \item $(U,X,\<\cdot,\cdot\>_\frakm,c)$ is an invariant contracting system on $M$.
    \end{enumerate}
\end{theorem}
\begin{proof}
    Let $\calA = \{A_1,\dots,A_m\}$ be an orthonormal basis for $\frakm$.
    By Lemma~\ref{lem:propPX}, since any $V\in T_pM$ can be written as $T_g\pi(v^jA_j^L(g))$ given a representative $g\in G$ with $p = gH$,
    \begin{align*}
        &\dangle{\nabla_{V}X, V}_p \leq c\dangle{V,V}_p \quad \forall p\in U,\,V\in T_pM \\
        &\iff \<\partial_\calA X(p)v,v\>_\calA \leq c\<v,v\>_\calA \quad \forall p\in U,\,v\in\R^m \\
        &\iff \left\<\partial_\calA X(p)\tfrac{v}{|v|_\calA},\tfrac{v}{|v|_\calA}\right\>_\calA \leq c \quad \forall p\in U,\,v\in\R^m\setminus\{0\} \\
        &\iff \left\<\partial_\calA X(p)w,w\right\>_\calA \leq c \quad \forall p\in U,\,w\in\R^m,|w|_\calA=1,
    \end{align*}
    completing the proof.
\end{proof}

As a consequence of Theorem~\ref{thm:inv_contr}, all of the desirable properties of a contracting system hold for an invariant contracting system~\cite[Prop. 2.5]{SimpsonBullo_ContrRiemann:2014}. For instance, an important Corollary (the usual defining property of a contracting system) is that distances between any two trajectories initialized in the contraction region shrink exponentially at rate $c$.

\begin{corollary}[{\cite[Thm. 2.3]{SimpsonBullo_ContrRiemann:2014}}] \label{cor:contr_bound}
    Let $M = G/H$ be a reductive homogeneous space with reductive decomposition $\frakg = \frakh \oplus \frakm$. 
    If $(U,X,\<\cdot,\cdot\>_\frakm,c)$ is an invariant contracting system on $M$, $U$ is a $K$-reachable forward $X$-invariant set, and $X$ is forward complete on $U$, for every $t\geq 0$,
    \begin{align*}
        d(\Phi_X^t(x_0),\Phi_X^t(y_0)) \leq Ke^{ct} d(x_0, y_0),
    \end{align*}
    for any initial conditions $x_0,y_0\in U$, where $d$ is the Riemannian distance associated to the invariant metric $\dangle{\cdot,\cdot}$.
\end{corollary}

\section{A NECESSARY CONDITION FOR A CONTRACTING SYSTEM}

As remarked in~\cite{SimpsonBullo_ContrRiemann:2014}, a contracting system on a manifold $M$ satisfying the conditions of Corollary~\ref{cor:contr_bound} necessarily implies that $U$ is a contractible subset, since a time reparameterization of the flow yields a homotopy of $U$ to the fixed equilibrium point. However, the converse is not generally true---given an arbitrary contractible subset $U$ of a manifold, there may not exist a vector field $X$ and a Riemannian metric $\dangle{\cdot,\cdot}$ such that $(U,X,\dangle{\cdot,\cdot},c)$ is a contracting system on $M$ for some $c<0$. 
In this section, we further investigate under what circumstances a subset $U$ can admit a contracting system. 
As we will show in Theorem~\ref{thm:not_global_contr} and Corollary~\ref{cor:natred_closed_geo}, a contraction region with an invariant metric cannot contain any loops generated by one-parameter subgroups from $\frakm$.

Before we present Theorem~\ref{thm:not_global_contr}, we first examine the circle group $\bbS$. 
While $\bbS$ is not a contractible set, we use the theory developed in the previous section to prove the following fact, motivating Theorem~\ref{thm:not_global_contr}.
\emph{Fact:} There is no globally $(U = \bbS)$ invariant contracting system on $\bbS$. 
For contradiction, suppose $(\bbS,X,\<\cdot,\cdot\>_\frakm,c)$ is an invariant contracting system on $\bbS$.
Let $A\in T_e\bbS$ be a unit vector, and let $x:\bbS\to\R$ be the curve such that $X = xA^L$.
Let $\hat{x}:\R\to\R$ be the periodic curve such that $\hat{x}(t) = x(t\text{ mod } 2\pi)$.
For any $t\in\R$,
\begin{align*}
    \partial X(t \text{ mod }2\pi) &= \calL_{A^L} x (t \text{ mod }2\pi) \\
    &= \lim_{h\to 0} \frac{x((t + h) \text{ mod }2\pi) - x(t \text{ mod } 2\pi)}{h} \\
    &= \lim_{h\to 0} \frac{\hat{x}(t + h) - \hat{x}(t)}{h} = \hat{x}'(t).
\end{align*}
Thus, $\<\hat{x}'(t)v,v\> = \hat{x}'(t) \<v,v\> \leq c\<v,v\>$, so $\hat{x}'(t) \leq c < 0$, for every $t\in\R$. But $\hat{x}$ is periodic, so this contradicts Rolle's theorem.
Figure~\ref{fig:failed_S} illustrates two failed attempts to construct an invariant contracting system on $\bbS$.

\begin{figure}
    \centering
    \includesvg[width=0.7\linewidth]{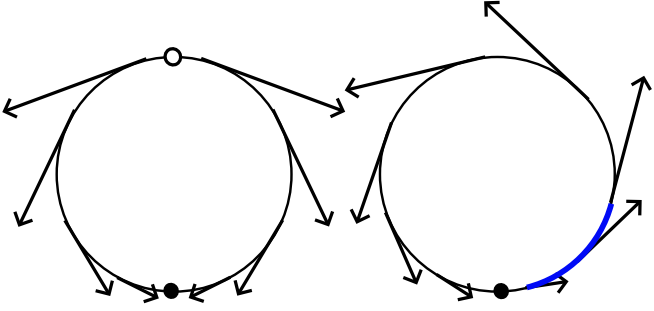}
    \caption{An illustration of two failed attempts to construct a globally contracting system on the circle equipped with an invariant metric. \textbf{Left:} The vector field is not continuous. \textbf{Right:} While the vector field is continuous, the distance indicated in blue increases.}
    \label{fig:failed_S}
    \vspace{-1em}
\end{figure}

In the following Theorem, we extend this intuition to show that the contraction region of an invariant contracting system on a reductive homogeneous space cannot contain the image of a circular one parameter subgroup of $\frakm$ through the group action $\lambda$.
In particular, we use a similar argument to show how the component of the vector field in the direction of such a subgroup cannot satisfy the contraction condition.

\begin{theorem} \label{thm:not_global_contr}
    Let $M = G/H$ be a reductive homogeneous space with reductive decomposition $\frakg = \frakh \oplus \frakm$, and let $(U,X,\<\cdot,\cdot\>_\frakm,c)$ be an invariant contracting system on $M$. For any $A\in\frakm$ generating a one parameter subgroup of $G$ isomorphic to the circle group $\bbS$, $U$ cannot contain the full set $\{\pi(g\exp(tA)) : t\in\R\}$ for any $g\in G$.
\end{theorem}
\begin{proof}
    For contradiction, suppose $(U,X,\<\cdot,\cdot\>_\frakm,c)$ is an invariant contracting system on $M$, let $A_1\in\frakm$ be such that $\{\exp(tA_1) : t\in\R\}\simeq\bbS$, implying the existence of $T>0$ such that $\exp(T A_1) = e$, and let $g\in G$ such that $\{\pi(g\exp(tA_1)) : t\in\R\} \subseteq U$.
    Let $\calA = \{A_1,A_2,\dots,A_m\}$ be a orthonormal basis for $\frakm$, $\olX\in\Gamma^\infty(\Hor(G))$ be the horizontal lift, and $\olX^i:G\to\R$ such that $\olX = \olX^iA^L_i$.
    Let $y = [1\ 0 \ \cdots \ 0]^T\in\R^m$.
    Consider the following function $f:\R\to\R$, where since $\partial X(\pi(g)) = \partial_\calA \olX(g)$,
    \begin{align*}
        f(t) &= \<\partial_\calA X^i_j(\pi(g\exp(tA_1)))y^jA_i,y^iA_i\> \\
        &= \<\partial_\calA \olX^i_1(g\exp(tA_1))A_i,A_1\>.
    \end{align*}
    For each $i=1,\dots,m$, let $\hat{x}^i:\R\to\R$ be the $T$-periodic curve $\hat{x}^i(t) = \olX^i(g\exp(tA_1))$.
    Unrolling Definition~\ref{def:horleftinvlin},
    \begin{align*}
        \partial_\calA \olX^i_1(g\exp(tA_1)) = (\calL_{A_1^L} \olX^i + \olX^k\alpha_{1k}^i) (g\exp(tA_1)).
    \end{align*}
    Since $\Phi_{A_1^L}^t(g) = g\exp(tA_1)$, 
    \begin{align*}
        \calL_{A_1^L} \olX^i(g\exp(tA_1)) 
        &= \lim_{h\to 0}\frac{\hat{x}^i(t+h) - \hat{x}^i(t)}{h} = (\hat{x}^i)'(t).
    \end{align*}
    Next, for any $k$, recall that $\alpha_{1k} = \alpha^i_{1k}A_i$, so $\alpha^i_{1k}$ is the $i$-th basis expansion in the orthonormal basis $\{A_1,\dots,A_m\}$. 
    Thus, $\alpha^1_{1k}$ is the orthogonal projection of $\alpha_{1k}$ onto $A_1$, \emph{i.e.},
    \begin{align*}
        &\alpha^1_{1k} = \<\alpha(A_1,A_k),A_1\> = \left\<\tfrac12 [A_1,A_k]_\frakm + U(A_1,A_k),A_1\right\> \\
        &= \tfrac12\<[A_1,A_k]_\frakm,A_1\> + \<U(A_1,A_k),A_1\> \\
        &= \tfrac12(\<[A_1,A_k]_\frakm,A_1\> + \<[A_1,A_1]_\frakm,A_k\>
        + \<A_1,[A_k,A_1]_\frakm\>),
    \end{align*}
    using the definitions of $\alpha$ and $U$ from~\eqref{eq:LCalpha} and \eqref{eq:LC_U}.
    Since $[A_1,A_1]_\frakm = 0$, $[A_k,A_1]_\frakm = -[A_1,A_k]_\frakm$, and the inner product is symmetric, $\alpha^1_{1k} = 0$.
    Finally, plugging into $f$, 
    \begin{align*}
        f(t) &= \<((\hat{x}^i)'(t) + \hat{x}^k(t) \alpha_{1k}^i)A_i, A_1\> \\
        &= ((\hat{x}^i)'(t) + \hat{x}^k(t) \alpha_{1k}^i) \delta_{i1} \\
        &= (\hat{x}^1)'(t) + \hat{x}^k(t) \alpha_{1k}^1 = (\hat{x}^1)'(t).
    \end{align*}
    By assumption, $f(t) = (\hat{x}^1)'(t) \leq c < 0$. But $\hat{x}^1(t)$ is periodic, contradicting Rolle's theorem.
\end{proof}

By~\cite[Prop. 23.28]{GallierQuantance_DiffGeoLGComp:2020}, the geodesics through $p = gH$ for a naturally reductive homogeneous space are exactly of the form $t\mapsto \pi(g\exp(tX))$ for $X\in\frakm$. Any $X\in\frakm$ generating a one parameter subgroup isomorphic to $\bbS$ therefore generates closed geodesics on $M$, leading to the following Corollary.

\begin{corollary} \label{cor:natred_closed_geo}
    Let $M$ be a naturally reductive homogeneous space (\emph{e.g.}, a Riemannian symmetric space). If $(U,X,\<\cdot,\cdot\>_\frakm,c)$ is an invariant contracting system, then $U$ contains no closed geodesics.
\end{corollary}

\begin{figure}
    \centering
    \includesvg[width=0.4\columnwidth]{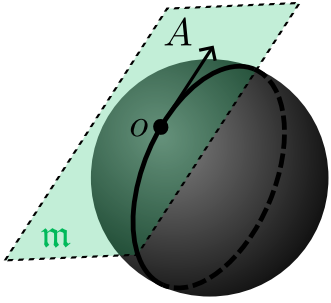}
    \caption{An illustration of a closed geodesic (great circle) passing through $o=eH$ on $\bbS^2 = SO(3)/SO(2)$ which is generated by an element $A\in\frakm$ as $\pi(\exp(tA))$. By Corollary~\ref{cor:natred_closed_geo}, a contraction region cannot contain any closed geodesic of this form; thus, no closed hemisphere is contained in the contraction region of any contracting system with the round metric.}
    \label{fig:global_S2}
    \vspace{-1em}
\end{figure}

While contractibility is a necessary topological condition for $U$ to satisfy, it does not account for the geometry of the space when equipped with a metric.
Theorem~\ref{thm:not_global_contr} and Corollary~\ref{cor:natred_closed_geo} partially bridge this gap, providing another necessary condition that accounts for the geometric properties of the space when equipped with an invariant metric.

\section{CASE STUDIES}

\begin{example} [Contraction on the sphere]
Consider a closed-loop single integrator on a sphere ($u\in\Gamma^\infty(T\bbS^2)$),
\begin{gather*}
    \dot{x} = u(x) \overset{\pi}{\iff} 
\begin{aligned}
    \dot{R} &= u^X(R) A_X^L(R) + u^Y(R) A_Y^L(R), \\ 
    &= T_e\ell_R(u^X(R)A_X + u^Y(R)A_Y),
\end{aligned} 
\end{gather*}
where we have identified $\bbS^2 = SO(3)/SO(2)$, with the round metric as Example~\ref{ex:homogeneous_examples}~\eqref{ex:n_sphere}, 
$\calA = \{A_X,A_Y,A_Z\}$ as the following basis for $\so(3)$,
\begin{align*}
    A_X = \begin{bsmallmatrix}
        0 & 0 & 0 \\
        0 & 0 & -1 \\
        0 & 1 & 0 
    \end{bsmallmatrix}, \
    A_Y = \begin{bsmallmatrix}
        0 & 0 & 1 \\
        0 & 0 & 0 \\
        -1 & 0 & 0 
    \end{bsmallmatrix}, \
    A_Z = \begin{bsmallmatrix}
        0 & -1 & 0 \\
        1 & 0 & 0 \\
        0 & 0 & 0 
    \end{bsmallmatrix},
\end{align*}
$\frakm = \operatorname{span}(A_X,A_Y)$, $\frakh = \operatorname{span}(A_Z)$, and $u^X,u^Y:SO(3)\to\R$.
The inner product $\<\cdot,\cdot\>_\frakm$ in the basis $\calA$ is the standard inner product $\<a,b\>_\calA = a^T b$.
For $i,j\in\{X,Y\}$, $4$ Lie derivative computations builds the $2\times 2$ matrix $(\partial_\calA \olu(R))_j^i = \calL_{A_j^L} u^i(R)$, since $\alpha = 0$ as $\bbS^2$ is a symmetric space. 
Thus, $\mu(\partial u(x)) = \mu_2(\partial_\calA \olu(R)) = \lambda_\text{max} (\tfrac12(\partial_\calA \olu(R) + \partial_\calA \olu(R)^T))$, for $x = \pi(R)$, and we can verify contraction for the round metric using a simple matrix measure computation.

Next, suppose $U\subseteq\bbS^2$ is a contraction region. Then Corollary~\ref{cor:natred_closed_geo} implies that no great circle (closed geodesic of the round metric) can be contained inside $U$. 
Thus, no closed hemisphere is contained in $U$. 
This is illustrated in Figure~\ref{fig:global_S2}.
\end{example}

\begin{example}[Reachable sets] \label{ex:SO3_reach}
One motivation for this work was to simplify contraction-based reachable sets for Lie groups using the distance estimate from Corollary~\ref{cor:contr_bound}. 
Consider the following control system on $SO(3)$ from~\cite{HarapanahalliCoogan_LieReach:2024},
\begin{align*} %
    \dot{R} &= X_u(R) = u_1A_X^L(R) + u_2A_Y^L(R) + u_3A_Z^L(R) \\
    &= u^1RA_X + u^2RA_Y + u^3RA_Z,
\end{align*}
with the basis $\calA = \{A_X,A_Y,A_Z\}$ for $\so(3)$ from the previous example
and the control input $u \in \R^3$, which we assume is fixed to a time-varying state independent mapping $\bfu:[0,\infty)\to\R^3$.
Consider the inner product on $\frakg$ where for every $\Theta,\Omega\in\frakg$, $\<\Theta,\Omega\>_\frakg = \frac12\operatorname{trace}(\Theta^T\Omega)$. This inner product is $\Ad_G$-invariant, so its corresponding invariant Riemannian metric is bi-invariant, and makes $SO(3)$ naturally reductive.
For any $R\in SO(3)$, we can compute $(\partial_\calA X_u)^i_j = \calL_{A_j^L} u^i = 0$,
for $i,j\in\{X,Y,Z\}$, thus, $\mu(\partial X(R)) = 0$ for any $R\in SO(3)$, so the system is nonexpansive ($c=0$)\footnote{Here, we are imprecise with how we handle the time-varying nature of the system. After fixing a map $\bfu:\R\to\R^3$, one can write the closed-loop system by adding time as a state and considering a vector field on $SO(3)\times \R$, with $\dot{t} = 1$. This system is indeed nonexpansive.}. While the system is not an invariant contracting system as in Definition~\ref{def:contr_sys} as $c$ is nonnegative, this still implies that for any two trajectories $R(t)$ and $R'(t)$,
\begin{align*}
    d(R(t),R'(t)) \leq d(R(0),R'(0)).
\end{align*}
Thus, to compute the reachable set from a metric ball initial set of the form $\calB_{R'_0}(r) = \{R\in SO(3) : d(R,R'_0) \leq r\}$, we simply need to simulate the trajectory from the center $R'(t)$, and the reachable set at time $t$ is guaranteed to be a subset of $\calB_{R'(t)}(r)$. This procedure is visualized in Figure~\ref{fig:SO3fig}, and provides far more accurate reachable set estimates when compared to~\cite{HarapanahalliCoogan_LieReach:2024}.

\emph{Discussion}\quad 
After choosing an orthonormal basis for $\frakg$, the analysis essentially reduces to analyzing the trivial system $\dot{x} = u$ in Euclidean space. Since we consider a fixed, open-loop control input, every two trajectories always remain the same distance apart. 
In the same way, on $SO(3)$, just one base trajectory needs to be simulated, and the rest are recoverable by left-translating from this trajectory.

\end{example}

\begin{figure}
    \centering
    \includegraphics[width=0.32\columnwidth]{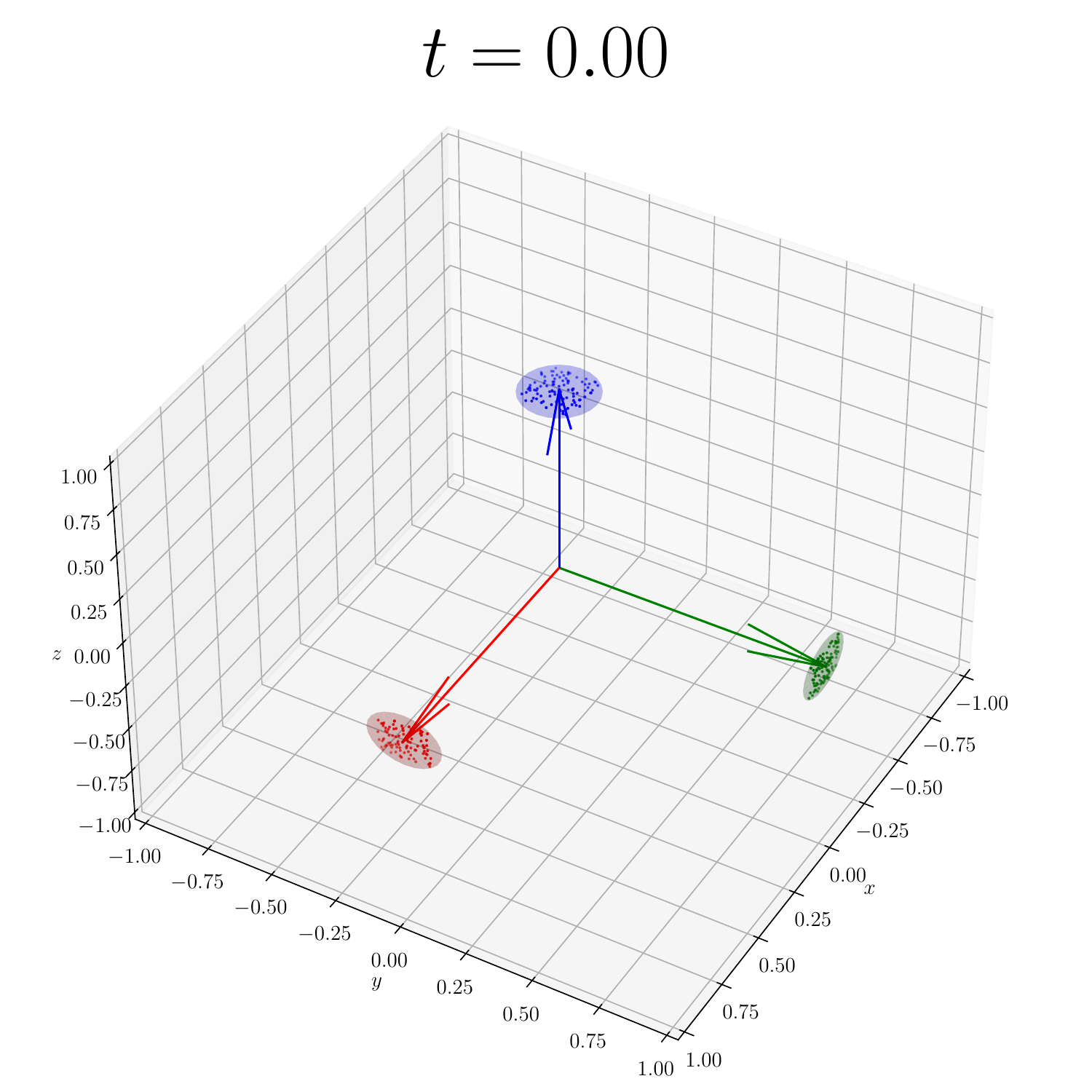}
    \includegraphics[width=0.32\columnwidth]{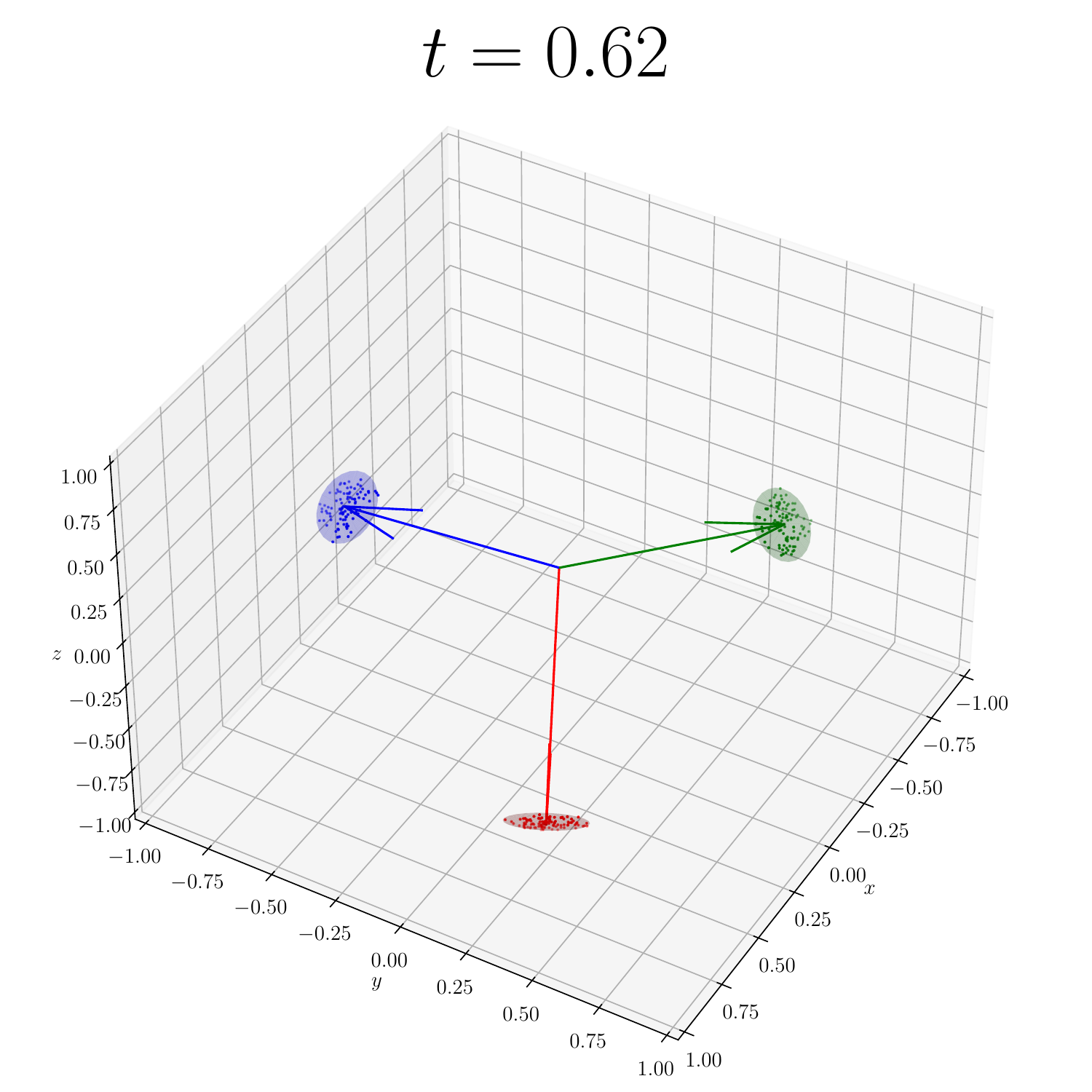}
    \includegraphics[width=0.32\columnwidth]{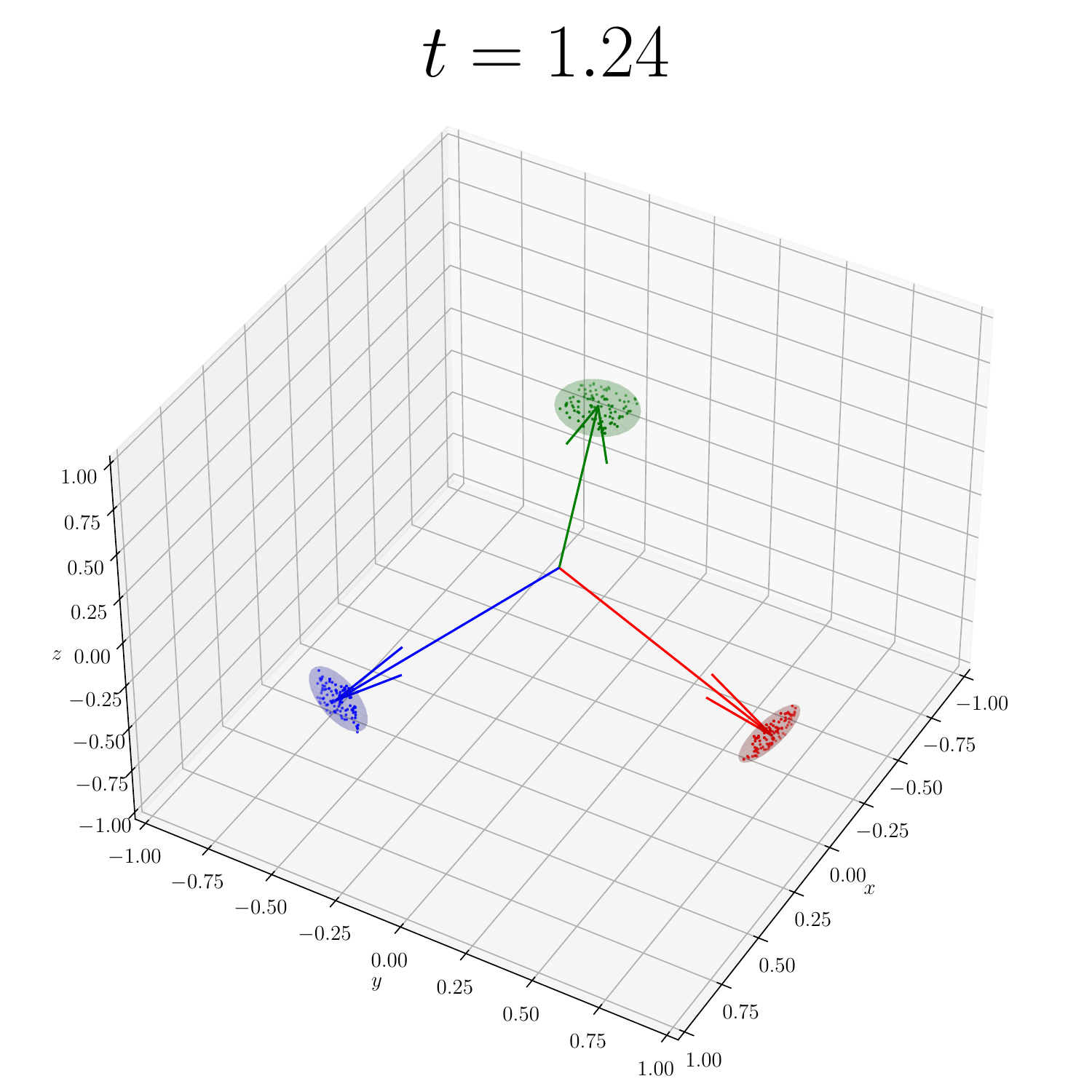}
    \includegraphics[width=0.32\columnwidth]{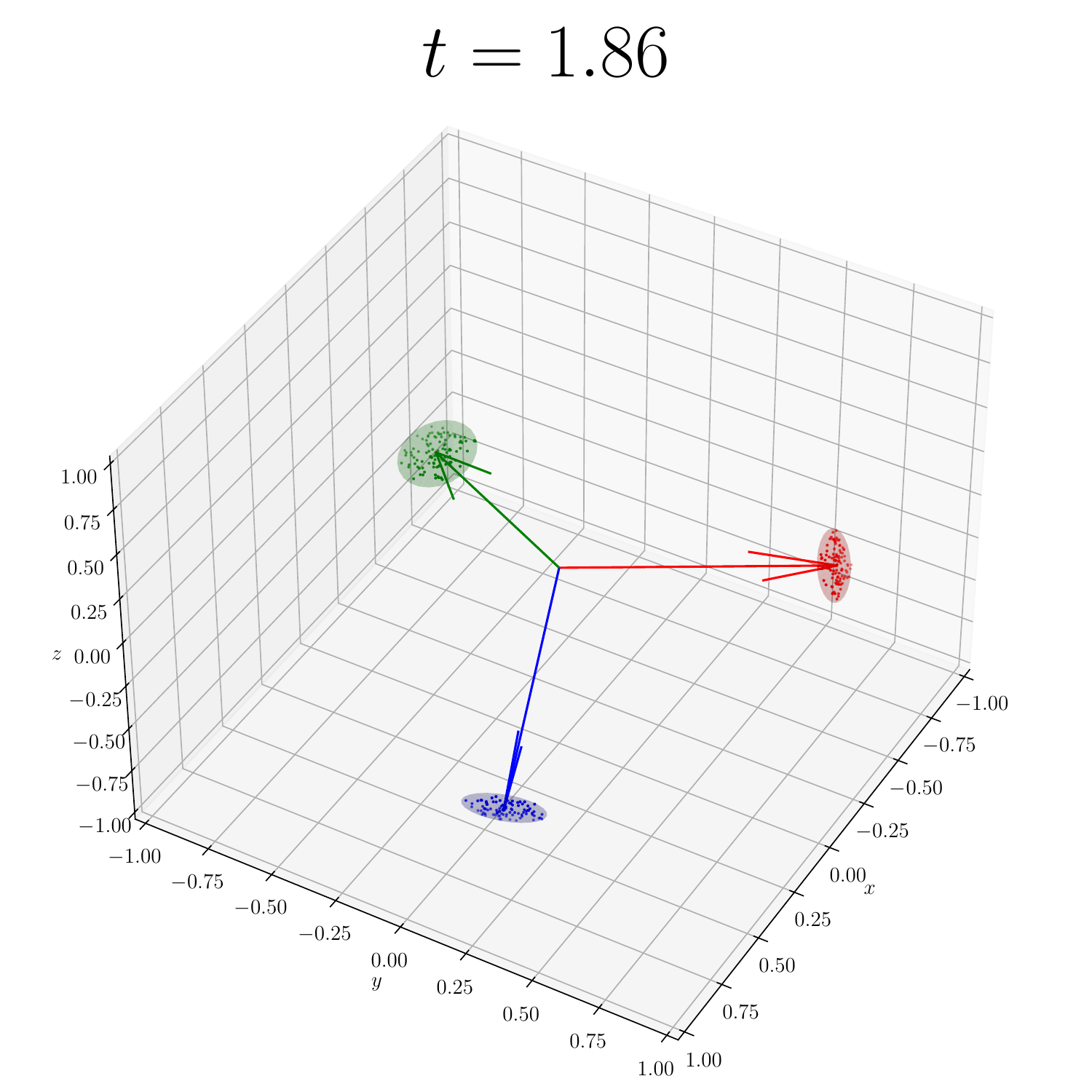}
    \includegraphics[width=0.32\columnwidth]{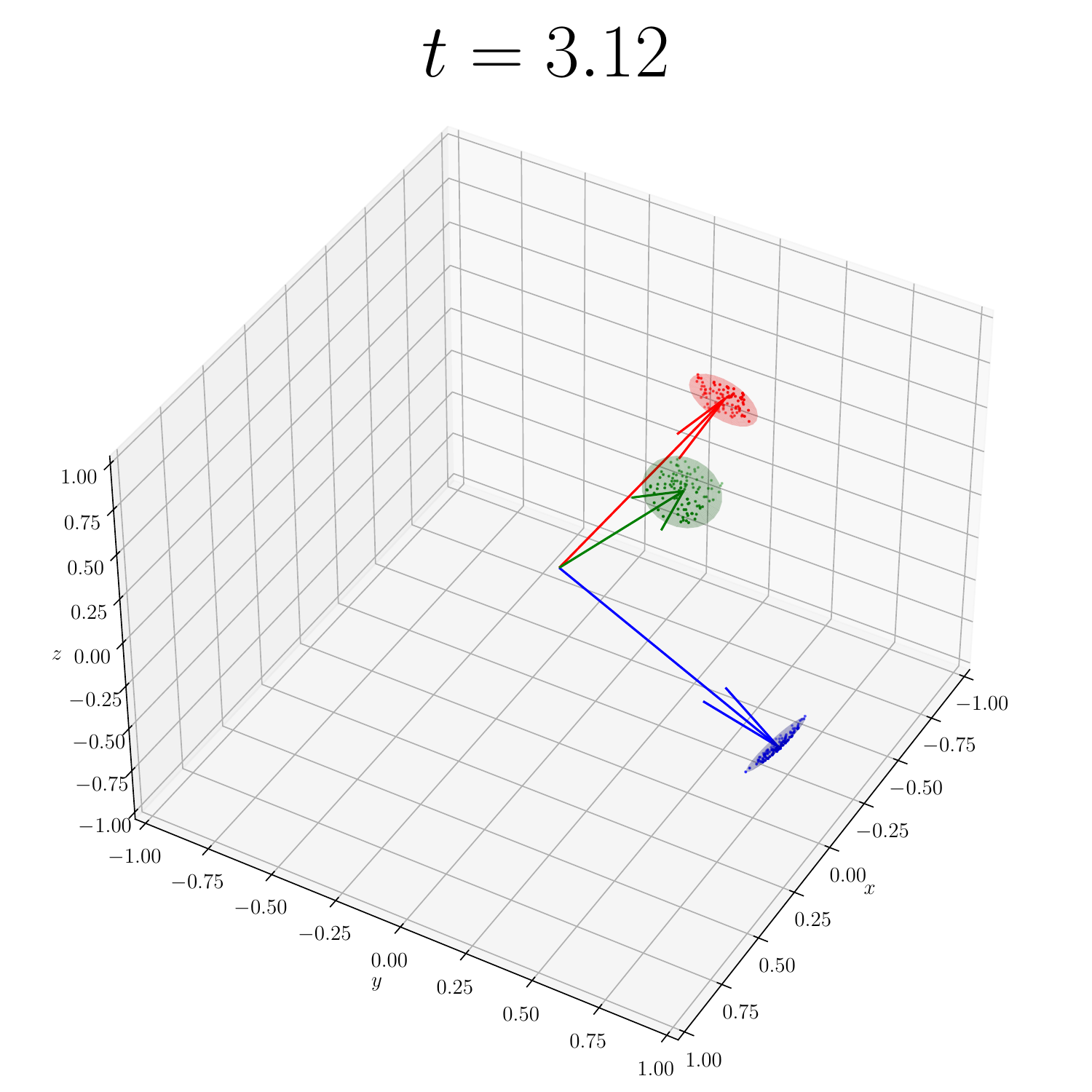}
    \includegraphics[width=0.32\columnwidth]{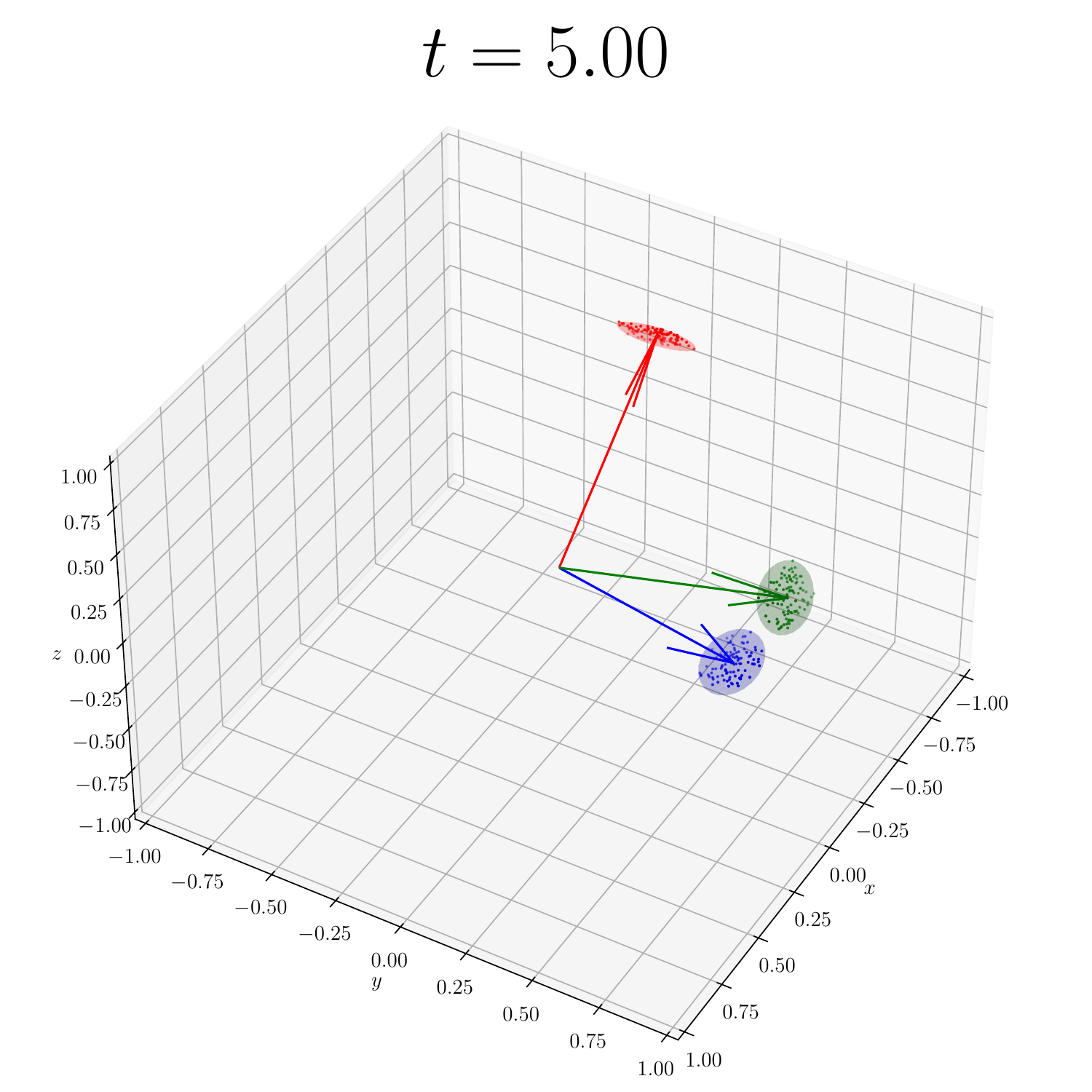}
    \caption{The reachable set for the orientation of a satellite evolving on $SO(3)$ is visualized as uncertainty sets around a coordinate axis. The system from Example~\ref{ex:SO3_reach} is simulated from the initial set $\{I\exp(X) : \sqrt{\<X,X\>_\frakg} \leq 0.1\}$, under the control mapping $\bfu:[0,5]\to\R^3$, $t\mapsto \bfu(t) = \left[\frac{5 - t}{5}, 1 - \left(\frac{t}{5}\right)^2, \sin\left(\frac{\pi t}{2}\right)\right]^T$. Since the system is nonexpansive, and there are no disturbances applied, the reachable set does not grow or shrink over time. 
    $100$ Monte Carlo simulations are also shown, empirically verifying their containment in the reachable set.
    }
    \label{fig:SO3fig}
    \vspace{-1em}
\end{figure}

\section{CONCLUSIONS}

In this work, we specialized existing results in coordinate-free contraction~\cite{SimpsonBullo_ContrRiemann:2014} to reductive homogeneous spaces equipped with an invariant Riemannian metric.
We showed how the Levi-Civita contraction condition can be simplified down to a standard matrix measure computation by lifting the problem to the Lie group. 
Using this result, we were able to provide a necessary condition for the existence of a contracting system, involving the underlying geometry rather than purely topological properties.
In future work, we plan to further investigate the connection between invariant contraction using the horizontal bundle $\Hor(G)$, horizontal contraction with respect to a horizontal Finsler-Lyapunov function on $G$~\cite{ForniSepulcre_DiffLyap:2013}, semi-contraction~\cite{JafarpourSVBullo_SemiContraction:2021}, and invariant Finsler metrics.
We plan to use this theory to study robustness in mechanical and robotic systems evolving on Lie groups.

\bibliographystyle{IEEEtran}
\bibliography{homogeneous}

\begin{thebibliography}{10}
\providecommand{\url}[1]{#1}
\csname url@samestyle\endcsname
\providecommand{\newblock}{\relax}
\providecommand{\bibinfo}[2]{#2}
\providecommand{\BIBentrySTDinterwordspacing}{\spaceskip=0pt\relax}
\providecommand{\BIBentryALTinterwordstretchfactor}{4}
\providecommand{\BIBentryALTinterwordspacing}{\spaceskip=\fontdimen2\font plus
\BIBentryALTinterwordstretchfactor\fontdimen3\font minus
  \fontdimen4\font\relax}
\providecommand{\BIBforeignlanguage}[2]{{%
\expandafter\ifx\csname l@#1\endcsname\relax
\typeout{** WARNING: IEEEtran.bst: No hyphenation pattern has been}%
\typeout{** loaded for the language `#1'. Using the pattern for}%
\typeout{** the default language instead.}%
\else
\language=\csname l@#1\endcsname
\fi
#2}}
\providecommand{\BIBdecl}{\relax}
\BIBdecl

\bibitem{AminzareSontag:2014}
Z.~Aminzare and E.~D. Sontag, ``Contraction methods for nonlinear systems: A
  brief introduction and some open problems,'' in \emph{53rd IEEE Conference on
  Decision and Control}, 2014, pp. 3835--3847.

\bibitem{Bullo:CTDS}
F.~Bullo, \emph{Contraction Theory for Dynamical Systems}, {1.2}~ed.\hskip 1em
  plus 0.5em minus 0.4em\relax Kindle Direct Publishing, 2024.

\bibitem{DavydovBullo:2024}
A.~Davydov and F.~Bullo, ``Perspectives on contractivity in control,
  optimization, and learning,'' \emph{IEEE Control Systems Letters}, vol.~8,
  pp. 2087--2098, 2024.

\bibitem{LohmillerSlotine_ContrNLSys:1998}
W.~Lohmiller and J.-J.~E. Slotine, ``On contraction analysis for non-linear
  systems,'' \emph{Automatica}, vol.~34, no.~6, pp. 683--696, 1998.

\bibitem{ManchesterSlotine:2017}
I.~R. Manchester and J.-J.~E. Slotine, ``Control contraction metrics: Convex
  and intrinsic criteria for nonlinear feedback design,'' \emph{IEEE
  Transactions on Automatic Control}, vol.~62, no.~6, pp. 3046--3053, 2017.

\bibitem{Sontag_CSWI:2010}
E.~D. Sontag, ``Contractive systems with inputs,'' in \emph{Perspectives in
  Mathematical System Theory, Control, and Signal Processing: A Festschrift in
  Honor of Yutaka Yamamoto on the Occasion of his 60th Birthday}.\hskip 1em
  plus 0.5em minus 0.4em\relax Springer, 2010, pp. 217--228.

\bibitem{DavydovJafarpourBullo_NonEuclidean:2022}
A.~Davydov, S.~Jafarpour, and F.~Bullo, ``Non-{E}uclidean contraction theory
  for robust nonlinear stability,'' \emph{IEEE Transactions on Automatic
  Control}, vol.~67, no.~12, pp. 6667--6681, 2022.

\bibitem{DeLellisRusso:2010}
P.~DeLellis, M.~di~Bernardo, and G.~Russo, ``On {QUAD}, {L}ipschitz, and
  contracting vector fields for consensus and synchronization of networks,''
  \emph{IEEE Transactions on Circuits and Systems I: Regular Papers}, vol.~58,
  no.~3, pp. 576--583, 2010.

\bibitem{RussoSontag:2012}
G.~Russo, M.~Di~Bernardo, and E.~D. Sontag, ``A contraction approach to the
  hierarchical analysis and design of networked systems,'' \emph{IEEE
  Transactions on Automatic Control}, vol.~58, no.~5, pp. 1328--1331, 2012.

\bibitem{WuYiManchester_CCMLieGroup:2024}
D.~Wu, B.~Yi, and I.~R. Manchester, ``Control contraction metrics on
  submanifolds,'' in \emph{2024 IEEE 63rd Conference on Decision and Control},
  2024, pp. 3735--3740.

\bibitem{Coogan:2019}
S.~Coogan, ``A contractive approach to separable {L}yapunov functions for
  monotone systems,'' \emph{Automatica}, vol. 106, pp. 349--357, 2019.

\bibitem{JafarpourBullo:2021}
S.~Jafarpour, A.~Davydov, A.~Proskurnikov, and F.~Bullo, ``Robust implicit
  networks via non-{E}uclidean contractions,'' in \emph{Advances in Neural
  Information Processing Systems}, vol.~34, 2021, pp. 9857--9868.

\bibitem{SimpsonBullo_ContrRiemann:2014}
J.~W. Simpson-Porco and F.~Bullo, ``Contraction theory on {R}iemannian
  manifolds,'' \emph{Systems \& Control Letters}, vol.~65, pp. 74--80, 2014.

\bibitem{ForniSepulcre_DiffLyap:2013}
F.~Forni and R.~Sepulchre, ``A differential {L}yapunov framework for
  contraction analysis,'' \emph{IEEE transactions on automatic control},
  vol.~59, no.~3, pp. 614--628, 2013.

\bibitem{WuDuan_GeoLyapCharISSFinsler:2022}
D.~Wu and G.-R. Duan, ``Further geometric and {L}yapunov characterizations of
  incrementally stable systems on {F}insler manifolds,'' \emph{IEEE
  Transactions on Automatic Control}, vol.~67, no.~10, pp. 5614--5621, 2022.

\bibitem{WuYiRantzer:2024}
D.~Wu, B.~Yi, and A.~Rantzer, ``Stability analysis of trajectories on manifolds
  with applications to observer and controller design,'' \emph{IEEE
  Transactions on Automatic Control}, vol.~69, no.~10, pp. 7104--7111, 2024.

\bibitem{Nomizu_InvAffConn:1954}
K.~Nomizu, ``Invariant affine connections on homogeneous spaces,''
  \emph{American Journal of Mathematics}, vol.~76, no.~1, pp. 33--65, 1954.

\bibitem{GallierQuantance_DiffGeoLGComp:2020}
J.~Gallier and J.~Quaintance, \emph{Differential Geometry and {L}ie Groups: A
  Computational Perspective}, ser. Geometry and Computing.\hskip 1em plus 0.5em
  minus 0.4em\relax Springer International Publishing, 2020.

\bibitem{Munthe-Kaas_RKMK:1999}
H.~Munthe-Kaas, ``High order {R}unge-{K}utta methods on manifolds,''
  \emph{Applied Numerical Mathematics}, vol.~29, no.~1, pp. 115--127, 1999.

\bibitem{MostajeranSepulchre_MonHomSpace:2018}
C.~Mostajeran and R.~Sepulchre, ``Monotonicity on homogeneous spaces,''
  \emph{Mathematics of Control, Signals, and Systems}, vol.~30, pp. 1--25,
  2018.

\bibitem{ForniSepulcre_DiffPosSys:2015}
F.~Forni and R.~Sepulchre, ``Differentially positive systems,'' \emph{IEEE
  Transactions on Automatic Control}, vol.~61, no.~2, pp. 346--359, 2015.

\bibitem{MostajeranSepulcre_InvDiffPos:2018}
C.~Mostajeran and R.~Sepulchre, ``Invariant differential positivity and
  consensus on {L}ie groups,'' \emph{IFAC-PapersOnLine}, vol.~49, no.~18, pp.
  630--635, 2016, 10th IFAC Symposium on Nonlinear Control Systems.

\bibitem{WeldeKumar:2024}
J.~Welde and V.~Kumar, ``Almost global asymptotic trajectory tracking for
  fully-actuated mechanical systems on homogeneous {R}iemannian manifolds,''
  \emph{IEEE Control Systems Letters}, vol.~8, pp. 724--729, 2024.

\bibitem{Schlarb_CovariantHomogeneous:2024}
M.~Schlarb, ``Covariant derivatives on homogeneous spaces: Horizontal lifts and
  parallel transport,'' \emph{The Journal of Geometric Analysis}, vol.~34,
  no.~5, p. 150, 2024.

\bibitem{HarapanahalliCoogan_LieReach:2024}
A.~Harapanahalli and S.~Coogan, ``Efficient reachable sets on {L}ie groups
  using {L}ie algebra monotonicity and tangent intervals,'' in \emph{2024 IEEE
  63rd Conference on Decision and Control}, 2024, pp. 695--702.

\bibitem{JafarpourSVBullo_SemiContraction:2021}
S.~Jafarpour, P.~Cisneros-Velarde, and F.~Bullo, ``Weak and semi-contraction
  for network systems and diffusively coupled oscillators,'' \emph{IEEE
  Transactions on Automatic Control}, vol.~67, no.~3, pp. 1285--1300, 2021.

\end{thebibliography}

\end{document}